\documentclass[12pt]{article}

\usepackage{enumerate}
\usepackage{mathtools, amssymb, latexsym, amsmath, amsfonts, amsthm, color}
\usepackage{array, fancyhdr, bm, listings}
\usepackage{algorithm,comment}
\usepackage{algorithmic}
\usepackage{subcaption}
\usepackage{thmtools, thm-restate}
\declaretheorem{theorem}

\theoremstyle{plain}
\newtheorem{lemma}[theorem]{Lemma}

\newtheorem{corollary}[theorem]{Corollary}

\theoremstyle{definition}

\DeclareMathOperator*{\argmin}{arg\,min}

\newcommand\mc{$\mathsf{MinCut}$ }
\newcommand\reg{\textsc{reg}}
\newcommand\regn{\textsc{regn}}
\newcommand\eregn{\textsc{eregn}}

\providecommand{\keywords}[1]{\textit{Keywords:} #1}
\title{Online MinCut: Competitive and Regret Analysis\footnote{This research was supported in part by the  DTIC contract FA8075-14-D-0002/0007}}

\date{}

\author{Avah Banerjee, Guoli Ding}
\newcommand{\FormatAuthor}[3]{
\begin{tabular}{c}
#1 \\ {\small\texttt{#2}} \\ {\small #3}
\end{tabular}
}
\author{
\begin{tabular}[h!]{lcr}
   \FormatAuthor{Avah Banerjee\footnote{Formerly Indranil Banerjee}}{banerjeeav@mst.edu}{Missouri S\&T\footnote{This research was done while the author was visiting LSU as a postdoc}}
&   \FormatAuthor{Guoli Ding}{ding@math.lsu.edu}{Louisiana State University}
\end{tabular}
}

\begin{document}

\maketitle

\begin{abstract}
In this paper we study the  
mincut problem  in the online setting. We consider two distinct models: A) competitive analysis and B) regret analysis.
In the competitive setting
we consider the vertex arrival model; whenever a new vertex arrives it's neighborhood with respect to the set of known vertices is revealed.
An online algorithm must make an irrevocable decision to determine the 
side of the cut that the vertex must belong to in order to minimize  the  size of the final cut. 
Various models are considered. 1) For classical and advice models we give tight bounds on the competitive ratio of deterministic algorithms.
2) Next we consider few semi-adversarial inputs: random order of arrival with adversarially generated and sparse graphs. 
3) Lastly we derive some structural properties of \mc-type problems with respect to greedy strategies.

Finally we consider a non-stationary regret setting with a variational budget $V_T$ and give tights bounds on the regret function. Specifically, we show that if $V_T$ is sublinear in $T$ (number of rounds) then there is a deterministic algorithm achieving a sublinear regret bound ($O(V_T)$). Further, this is optimal, even if randomization is allowed.
\end{abstract}


\tableofcontents

\section{Introduction}
In the first part of the paper we consider the online \mc problem under  competitive analysis. In the second part we use a regret model. Lastly, we  study a structural property of \mc and other submodular functions with respect to a greedy order of the ground set.
\subsection{Competitive Analysis}\label{sec: comp models}
Let $\Gamma$ be a (possibly infinite) graph.
An online graph $G=(V,E)$ is a finite subgraph of $\Gamma$ and there is a total order $\pi$ on $V(G)$ or (and) $E(G)$.  We assume $|V(G)| = n$ and $|E(G)| = m$. 
Sometimes $\Gamma$ is not mentioned when we describe a problem because $G$ is allowed to be any finite graph and $\Gamma$ is the disjoint union of all finite graphs (and thus there is no need to mention $\Gamma$).  
In the {\it vertex arrival} model, vertices of $G$ are revealed one at a time according to $\pi$, along with its neighbors in the current set of revealed vertices.
In the {\it edge arrival} model, vertices of $G$ are known and edges arrive one at a time according to $\pi$.
We do not explicitly consider the edge arrival model in this paper.
However some of our results in the vertex arrival model can be extended to the latter setting without much effort. 

Next we introduce some standard notions in competitive analysis \cite{borodin2005online}.
However we frame our discussions in terms of online graph problems.
Let $\mathsf P$ be some graph optimization problem  and let $opt_\mathsf P(G)$ be the optimal value of $\mathsf P$ when the input is $G$.
It is also known as the offline optimal.
Let $\mathbf A_\mathsf P(G, \pi)$ be the output computed by some online algorithm $\mathbf A$ given the ordering $\pi$.
We use competitive analysis to measure the relative performance of $\mathbf A$ with respect to the offline minimum.
Specifically, we say that $\mathbf A$ is $c$-{\it competitive} if for every $G$, 
$$\max_\pi \mathbf A_\mathsf P(G,\pi) \le c\ opt_\mathsf P(G) + d$$ 
for some constant $d \ge 0$.
If $d = 0$ then the algorithm is said to be {\it strictly} $c$-competitive.
The smallest $c$ for which an algorithm is (strictly) $c$-competitive is known as the {\it (strict) competitive ratio}. An algorithm is said to be \emph{competitive} if it is 1-competitive.
For maximization problems the competitive ratio is defined in a similar manner.
It is important to note that the constant $d$ must be independent of $G$ but may depend on $\mathbf{A}$ and $\mathsf P$. We omit the subscript $\mathsf{P}$ whenever the context is clear.

In the adversary model the input $G$ and its order of arrival (be it vertex or edge arrival) are determined by an adversary. 
At each step the algorithm makes an irrevocable decision based on the part of the input seen so far. No other knowledge about the input is known to the algorithm in advance (not even the length). 
This model is used for both deterministic and randomized algorithms \cite{borodin2005online}.
In the deterministic setting, the adversary knows the algorithm (also referred to as the online player) in advance.
For every input sequence the adversary knows the sequence of actions performed by the algorithm.
Hence it is often assumed that adversary creates the entire input then feeds the online algorithm one piece at a time.
However, in the case of randomized algorithms the notion of an adversary is a bit more complex.
Due to usage of random bits, the behaviour of an randomized online algorithm may differ in each run even with the same input sequence.
Informally, the power of an adversary depends on whether they are allowed to look at the current state of the online algorithm before deciding the next input.

Some online models can be considered semi-adversarial or non-adversarial.
They are often characterized in terms of an weak adversary.
For example, in the \emph{advice model}, the online algorithm is supplied with additional information by a benevolent oracle.
Interested readers can refer to \cite{hiller2012probabilistic,dorrigiv2010alternative, koutsoupias2000beyond, dehghani2017stochastic, soto2013matroid} for a more detail overview of these models.
Some of the more well known models  are random-order model (for the matroid secretary problem), diffuse adversary (for paging), Markov process (for paging) etc.
Resource augmentation based models, where the adversary is made weak by giving more ``resources" to the online algorithm can also be thought of as semi-adversarial.
A good example is the $(h,k)$-server problem ($h \le k)$  \cite{bansal2019h}.
Here $h$ is the number of server the adversary is allowed to used to process the requests they generate.
These models are an important alternative to the adversarial models as they attempt to represent real world situations more accurately.
We consider two such models: 1) when the input set is either restricted or is semi-random/ random 2) the algorithm has access to an oracle that knows the  input in advance (among other knowledge).

\subsubsection{Semi-adversarial Inputs}\label{secc: raand order}
In the context of online graph problems, we look at a relevant semi-adversarial model. The arrival order of vertices are chosen uniformly at random.
In this setting we consider two situation:
(1) The graph $G$ is adversarially generated
(2) The graph $G$  comes from a particular family of graphs which is known to the algorithm in advance.
In particular we look at  sparse graphs.

In the random order model we want to determine the the competitive ratio in terms of the expected value of the solution determined by the algorithm.
That is,

$$\mathbb E[\mathbf A_\mathsf P(G,\pi)] =  \frac{1}{n!}\sum_{\pi}\mathbf A_\mathsf P(G,\pi) \le c\ opt_\mathsf P(G) + d$$

The above expectation is over the random permutation $\pi$ and possibly over the random choices made by $\mathbf{A}_P$ .
Since $G$ is not random, the optimal value is not a random variable.
If the input graph is selected according to some distribution then we have use $\mathbb{E}[opt_\mathsf P(G)]$ instead.


\subsubsection{Advice Model}
Advice in the context of online computation is a
model where some information about the future inputs are available to the algorithm.
Its inception is somewhat recent \cite{boyar2016online,dobrev2009measuring}.
The informal idea is as follows. 
The online algorithm is given access to a friendly  oracle which knows the input in advance.
The oracle is assumed to have unlimited computational power.
The algorithm is allowed to ask arbitrary questions to this oracle  at any stage  of the computation.
We do not care about the nature of information received rather the amount, in terms of the number of bits.
This quantity is known as the advice complexity of the algorithm.
Given some online problem $\mathsf P$ we want to determine the lower (upper) bound of the amount of advice needed by any (some)  algorithm to achieve a certain competitive ratio.
This model have been shown to be useful in proving certain lower bounds for online problems.

There are various flavors of advice models, which are more or less equivalent. 
The model we use here is a variant of the \emph{tape model} \cite{hromkovivc2010information}.
Let $\mathsf P$ be some online graph minimization problem.
Let $\mathbf A_\mathsf P^{adv}$ be an algorithm solving $\mathsf P$ which has access to an advice string $adv$.
We say $\mathbf A_\mathsf P^{adv}$ is $c$-competitive with advice complexity $b$ for $\mathsf P$ if there is an advice string $adv$ of size at most $b$ such that:

$$\max_\pi \mathbf A_\mathsf P^{adv}(G,\pi) \le c\ opt_\mathsf P(G) + d$$
Where $d$ is some constant independent of the size of $G$.
The advice complexity $b$ can be a function of the size of $G$, however it is not dependent on $G$ itself.
In the above definition we implicitly assume the length of the advice string is known to the algorithm.
Otherwise we may assume advice strings are self delimiting adding to a $O(\log b)$ overhead.

\subsection{Regret Analysis}
In competitive analysis we are interested in comparing the 
optimal offline solution for the full input sequence with a solution obtained incrementally by an online algorithm making a sequence of irrevocable decisions.
Regret analysis, in contrast, is often used in domain of online decision making and learning.
At a high level at each step we play a move from a set of feasible actions and we  receive a feedback\cite{hazan2012online}.
Depending on the problem this feedback may directly or indirectly specify the loss we incur after playing the action.
In the regret setting we are interested in measuring the total loss relative to an  algorithm whose moves are determined in hindsight\footnote{This algorithm need not be optimal. Different assumptions leads to different notions of regret \cite{jegelka2011online}}.
Online minimization problems are a natural class of problems to study in the regret settings.
At each time step $t$ an online algorithm plays a feasible solution $x_t$. Then it receives a feedback $f_t(x_t)$.
Depending on the model the feed back mechanism can be explicit or implicit.
For example instead of the function value we may receive the gradient of the function at $x_t$.
In the context of the \mc problem, at each time step, after the online algorithm chooses a cut the adversary will supply a new weight function $w_t$.
We consider the full information setting where this weight function is fully specified after the algorithm has made its choice.
Next we define our regret measure.
Early studies on regret based learning primarily focused on stationary regret, which is defined as follows \cite{kalai2005efficient, hazan2012online, jegelka2011online}.
\begin{align}
    \reg(\mathbf{A}) = \sum_{t=1}^T{f_t(x_t)}- \min_{x \in \mathcal{X}} \sum_{t=1}^T{f_t(x)}
\end{align}
Here $\mathcal{X}$ is the feasible set and $x^* = \argmin_{x \in \mathcal{X}} \sum_{t=0}^T{f_t(x)}$ is a solution that minimizes the cumulative cost.
Intuitively, this measures the cost $\mathbf{A}$ incurs while choosing a different solution at each step instead of choosing a single solution in hindsight.
One of the goal in this model is to determine for a particular problem if  there is any \emph{Hannan-consistent}  algorithm. Such an algorithm exhibits a sublinear regret in $T$, the number of rounds. 
This implies that eventually the solutions obtained by the algorithm ``converge" to the best compromised offline solution.
The online MinCut problem in this setting can be thought of as a constrained convex (in fact linear) optimization on the space of the characteristic vectors corresponding to the cuts.
It was shown in \cite{lugosi2009online,jegelka2011online} that the online minimum cut problem has sublinear regret under the stationary regret measure as above, even if the feedback function is submodular.

However, the stationary regret measure can me limiting.
For problems like the \mc, if the weight function changes at each step, it is easy to see that no single cut will be close the minimum value for each of the individual weight functions.
Indeed, many recent studies have focused on various forms of adaptive or non-stationary regret measure \cite{gao2018online, besbes2015non, roy2019multi}.
In the non-stationary setting we compare the online decisions against the best decisions for the corresponding time steps:

\begin{align}\label{eq:non  reg}
    \regn(\mathbf{A}) = \sum_{t=0}^T{f_t(x_t)}- \sum_{t=0}^T {f_t(x_t^*)}
\end{align}
where $x^*_t = \argmin_{x \in \mathcal{X}}{f_t(x)}$.
Clearly, $\reg(\mathbf{A}) \le \regn(\mathbf{A})$. In fact without any restrictions on the input sequence $(f_1,\ldots,f_T)$ the regret will be linear\cite{besbes2015non}.
Generally some variational bound is proposed for the input sequence (also known as the variational budget):
\begin{align}\label{eq: variation}
    \mathcal{F}_T = \left\{(f_1,\ldots,f_T)\mid \sum_{t=1}^{T-1}{||f_t - f_{t+1}|| \le V_T}\right\}
\end{align}
where the norm $||\cdot||$ can be realized by different metrics (usually it is the Minkowski norm).   $\mathcal{F}_T$ is the collection of such input sequences; known as the uncertainty set.
For some restricted classes of non-convex functions there have been some promising results recently with $O(\sqrt{T+TV_T})$ regret in the non-stationary setting\cite{gao2018online}. This is sublinear if $V_T$ is.
The regret model is usually considered with respect to a randomized algorithm. 
We can think of this as a repeated game where the online player chooses a mixed strategy and the adversary chooses a feedback.
Equation \ref{eq:non  reg} can be modified in the randomized setting as follows:
\begin{align}\label{eq:rand non  reg}
    \eregn(\mathbf{A}) = \mathbb{E}\left[\sum_{t=0}^T{f_t(x_t)}- \sum_{t=0}^T {f_t(x_t^*)}\right]
\end{align}
where the expectation is taken over the random variables $x_t$ and possibly over any randomness in the sequence $f_t$.


\subsection{Greedy property of online submodular functions}

For many online problems the classical worst case model yields pessimistic results.
A review of some well known alternatives can be found here \cite{hiller2012probabilistic,dorrigiv2010alternative}.
However there are online problems, particularly in the minimization setting, where these model fail to distinguish the hardness of these problems. 
Graph problems, such as finding the mincut, min-degree, minimum spanning tree, minimum dominating set (discussed later) are good examples of online problems which are considered hard even with many beyond-worst-case measures.

With this in mind we look at the following measure to evaluate the hardness of some online problems on graphs. 
An extension of this idea can also be used to compare the ``robustness'' of different online algorithms even if their worst case performance are indistinguishable. 
At a high level we classify problems based on whether there is a ``good ordering" of the inputs for every possible input graphs such that we can always find an optimal output using a fixed (necessarily greedy) strategy.
Let $\mathbf A$ be an online algorithm for a graph minimization problem $\mathsf P$ under the vertex arrival model. The input of $\mathbf A$ is a permutation $\pi(V(G)) =v_1,v_2, \cdots, v_n$ of a graph $G$. To measure the performance of $\mathbf A$, we usually consider $\max \mathbf A(G, \pi)$ over all permutations $\pi$ in the adversary model. Similarly for the random order model with adversarially generated input we are interested in the average $\sum \mathbf A(G,\pi)/n!$. Along this line, one natural question we may ask is, what is $\min \mathbf A(G,\pi)$, over all $\pi$? 

For many problems $\mathsf P$, it is easy to construct $\mathbf A$ such that $\min_\pi \mathbf A(G,\pi) = opt (G)$ holds for all graphs $G$. For instance, if $\mathsf{P} =\mathsf {MinCov}$ is the minimum vertex-cover problem, then it is clear that the following $\mathbf A$ satisfies the requirement: placing $v_i, v_{i+1},...,v_n$ in the cover, where $i$ is the largest index such that $v_1,v_2,...,v_{i-1}$ is an independent set. To see that $\min_\pi \mathbf A(G,\pi) = opt (G)$ we only need to take a maximum independent set $I$ and define $\pi$ to be a permutation that first lists all vertices of $I$ and then vertices of $V\backslash I$.

However, there are also problems for which no matter what $\mathbf A$ is, $\min_\pi \mathbf A(G,\pi)$ is different from $opt(G)$ for at least one graph $G$. For instance, consider the minimum domination problem $\mathsf{P} = \mathsf {MinDom}$: find a smallest set $D$ of vertices of $G$ such that every vertex outside $D$ is adjacent to at least one vertex inside $D$. Then $\mathsf P$ is such a problem. Suppose otherwise that $\mathbf A$ satisfies $\min_\pi \mathbf A(G,\pi)=opt(G)$ for all $G$. Then 
\begin{enumerate} 
\item $\mathbf A$ must place $v_1$ in $D$ because $G$ might have only one vertex. In general, if $v_1,...,v_k$ is independent then $\mathbf A$ has to place all of them in $D$.  
\item If $v_1$ is adjacent to $v_2$ then $\mathbf A$ must place $v_2$ outside $D$ because $G$ might be $K_2$. In general, if $v_1$ is adjacent to $v_2,...,v_k$ and $\{v_2,...,v_k\}$ is independent then $\mathbf A$ must place $v_2,...,v_k$ outside $D$ since $G$ might be $K_{1,k-1}$.
\end{enumerate} 
Now let $H$ be the tree with five edges $13, 23, 34, 45$ and $46$. Then $H$ has a unique minimum dominating set $\{3,4\}$. If $\pi$ is a permutation so that $\mathbf A(H,\pi)=\{v_3,v_4\}$, by (1) above we may assume $v_1=3$. Then (2) implies a contradiction. So, no matter what $\mathbf A$ is, $\mathbf A(H,\pi)\ne opt(H)$ for all $\pi$.

The above two examples show the two extremes concerning $\min_\pi \mathbf A(G, \pi)$. In this paper we establish that if $\mathsf P$ is $\mathsf{MinCut}$ then there exists an algorithm with $\min_\pi \mathbf A(G, \pi)=opt(G)$ holds for all $G$. We extend our results to other graph optimization problems such as online maxcut and sub-modular function maximization \cite{streeter2009online}.

It is worth pointing out that the importance of our result is not the construction of an  algorithm $\mathbf A$. What's important about our result is that it reveals the structural difference between $\mathsf{MinCut}$ and problems like $\mathsf{MinDom}$. It illustrates that at least one optimal solution of $\mathsf{MinCut}$ can be identified in the online fashion. The result is more about the structure of $\mathsf{MinCut}$ than about algorithm $\mathbf A$.



\section{Results Under Competitive Analysis}
\subsection{Problem Definition and Notations}

Let $G=(V,E)$ be a graph. For any disjoint subsets $X,Y\subseteq V$, we denote by $E(X,Y)$ the set of all edges of $G$ that are between $X$ and $Y$. A {\it partition} of $V$ is a pair $(X,Y)$ of disjoint subsets of $V$ with $X\cup Y=V$. A {\it cut} of $G$ is a set $C\subseteq E$ that can be expressed as $E(X,Y)$ for a partition $(X,Y)$ of $V$ with $X\ne\emptyset \ne Y$. Note that every graph with two or more vertices must have at least one cut. 

The {\it minimum cut} problem ($\mathsf{MinCut}$) is to minimize $|C|$ over all cuts $C$ of $G$. Note that the minimum is finite for all $G$ with two or more vertices, and the minimum is $\infty$ if $|V(G)|=1$ since we are minimizing over the empty set. For a graph with a positive edge weights $w : E \to \mathbb R^+$ the problem ($\mathsf{MinCut}^+$) is to minimize $w(C)$, where $w(C) = \sum_{e\in C} w(e)$. 

Let $\cal G$ be a class of graphs. All graphs considered in the paper are simple. 
By $\mathsf{MinCut}[\cal G]$ we denote the problem $\mathsf{MinCut}$ with its input limited to graphs in $\cal G$.
According to our definition (in section~\ref{sec: comp models}), an online algorithm $\mathbf A$ for $\mathsf{MinCut}[\cal G]$ is called $c$-{\it competitive} if there exists a constant $d$ (which may depend on $\cal G$) such that for all $G \in {\cal G}$,
$$\max_\pi\mathbf A(G, \pi)\le c\cdot opt(G)+d.$$ 

\noindent For any integer $k\ge0$, let $\mathcal G_k$ denote the class of $k$-edge-connected graphs. Equivalently, $\mathcal G_k$ consists of all graphs $G$ with $opt(G)\ge k$. 
In addition, every graph in $\mathcal G_k$ has at least $k+1$ vertices.  
\noindent We use $\mathcal{G}_{(n)}$ to denote an infinite collection of graphs. The  collection contains graphs of size $n$ whenever $n$ is sufficiently large.

We consider the minimum cut problem in the advice model as follows.
The input, which is generated by the adversary,  is a graph $G$ together with a total order $\pi$ on its vertices. We denote the vertices, under $\pi$, by $v_1,v_2,...,v_n$ throughout our discussion. 
A partial input sequence $(v_1,\ldots,v_i)$ is termed as a prefix sequence.
By symmetry we assume $v_1\in X$. The algorithm may choose to ask questions even before $v_1$ is revealed. Since the placement of $v_1$ is fixed, it does not matter if these questions are asked before or after $v_1$ is revealed. To be consistent with all other steps, we assume that $\mathbf A$ does not ask anything before $v_1$ is revealed. So the process goes as follows: 

\noindent Step 1: $v_1$ is revealed and is placed in $X$. \\ 
Step 2: $v_2$ is revealed, then $\mathbf A$ asks a question and gets an answer, then $v_2$ is placed in $X$ or $Y$. \\ 
Step 3: $v_3$ is revealed, then $\mathbf A$ asks a question and gets an answer and so on.

At the $i$th $(i>1$) step of the computation, $\mathbf A$ has placed $v_1,...,v_{i-1}$ in $X$ or $Y$ already, $v_i$ is just revealed, and $\mathbf A$ needs to decide where to place $v_i$. At this point, $\mathbf A$ will ask a question about $(G,\pi)$, with the knowledge of $G[v_1,...,v_i]$ (the subgraph of $G$ induced on $v_1,...,v_i$) and possibly other information about $(G,\pi)$ that was obtained by $\mathbf A$ from the previous inquires. We define $\Gamma_i$ as the collection of potential inputs $G$ after seeing the first $i$ vertices.
A partition $(X_i,Y_i)$ of $\{v_1,...,v_i\}$ is called {\it extendable} if it can be extended into an optimal solution.

\subsection{Related Work}
To the best of our knowledge  $\mathsf{MinCut}$ and its other siblings (like min-bisection)  have not been studied in the competitive analysis setting.
In contrast there have been few results related to $\mathsf{MaxCut}$.
The folklore randomized $2$-approximation for the offline $\mathsf{MaxCut}$ also works in the online setting. 
In \cite{bar2012online} authors gave a almost tight bound of $3\sqrt{3}/2$ for the competitive ratio of the maximum directed cut problem under the vertex arrival model.

Few other studies have been made for online minimization problems on graphs.
Two important problems in this area are online minimum spanning tree  and coloring\cite{halldorsson1994lower, forivsek2012advice, bartal1996lower}. 
For the minimum spanning tree problem generally the edge arrival model is used. 
In \cite{remy2007online} authors study this problem when the edge weights are selected uniformly at random from $[0,1]$.
More recently this problem has been studied in the advice setting \cite{bianchi2018online}.

\subsection{Adversarial Input and Advice Complexity}
\begin{restatable}{theorem}{mincutclass}
\label{thm: classic}
(i) Let $\mathbf A$ be an online algorithm for $\mathsf{MinCut}[{\cal G}_{k}]$, where $\mathbf A$ knows $n$ in advance. Then the following hold. \\ 
\indent (a) If $k=0$ then $\mathbf A$ is not $c$-competitive for any $c$.\\ 
\indent (b) If $k\ge1$ and $\mathbf A$ is $c$-competitive then $c\ge \frac{n-p}{k}$ for some $p\ge1$. \\ 
(ii) Suppose $k\ge1$. Then there exists an online algorithm $\mathbf A$ for $\mathsf{MinCut}[{\cal G}_{k}]$, where $\mathbf A$ does not know $n$ in advance, such that $\mathbf A$ is $\frac{n-p}{k}$-competitive for all $p\ge1$.
\end{restatable}

\begin{proof}
\noindent (a) Let $G$ be obtained from $K_{n-1}\backslash e$ (where $n\ge4$ and $e=xy$) by adding an isolated vertex $z$. The adversary first reveal two nonadjacent vertices $v_1,v_2$. If $\mathbf A$ places $v_1,v_2$ in the same part of the partition, then the adversary can declare $v_1=x$ and $v_2=z$. In this case $\mathbf A(G,\pi)\ge n-3$. If $\mathbf A$ places $v_1,v_2$ in different parts of the partition then the adversary can declare $v_1=x$ and $v_2=y$. In this case $\mathbf A(G,\pi)\ge n-3$ holds again. If $\mathbf A$ is $c$-competitive, then there exists a number $d$ independent of $G$ and $\pi$ such that $\mathbf A(G,\pi)\le c\cdot opt(G)+d$ holds for all our $G$ and $\pi$. It follows that $n-3\le c\cdot 0+d$ holds for all $n\ge4$. This is impossible and thus $\mathbf A$ is not $c$-competitive for any $c$.
 
(b) Since $\mathbf A$ is $c$-competitive, there exists a constant $d$ satisfying $\mathbf A(G,\pi)\le c\cdot opt(G)+d$ for all $G\in\mathcal G_k$ and all $\pi$ on $G$. Without loss of generality, we assume $d\ge0$. Let $G$ be obtained from $K_{n-1}$ ($n>k$) by adding a new vertex $z$ and joining it to $k$ vertices of $K_{n-1}$. Then $opt(G)=k$ and thus $G$ belongs to $\mathcal G_k$. The adversary first reveal two adjacent vertices $v_1,v_2$. If $\mathbf A$ places $v_1,v_2$ in the same part of the partition, then the adversary can declare $v_1=z$. In this case $\mathbf A(G,\pi)\ge n-2$. If $\mathbf A$ places $v_1,v_2$ in different parts of the partition then the adversary can declare that neither $v_1$ nor $v_2$ is $z$. In this case $\mathbf A(G,\pi)\ge n-2$ holds again. Let $p=d+2$. Then $p\ge1$. In addition, $n-2\le c\cdot k+d$, implying $c\ge \frac{n-p}{k}$, as required. 

(ii) Let $\mathbf A$ be the following simple online algorithm for $\mathsf{MinCut}[{\cal G}_{k}]$: placing the first revealed vertex in the first part of the partition and all other vertices in the second part of the partition. Note that $\mathbf A$ does not need to to know $|G|$ in advance. We now prove that $\mathbf A$ is $\frac{n-p}{k}$-competitive for all $p\ge1$. To do so, we choose $d=(p-1) + \frac{(p-1)^2}{k}$ and we show that $\mathbf A(G,\pi)\le \frac{n-p}{k} \cdot opt(G) +d$ holds for all $G\in \mathcal G_k$ and all $\pi$ on $G$, which will prove (ii). We consider two cases.

If $n\le p$ then $\mathbf A(G,\pi)\le n-1 \le p-1 \le (p-1) + \frac{(p-1)^2}{k} + \frac{n-p}{k}\cdot (n-1) = \frac{n-p}{k}\cdot (n-1) + d\le \frac{n-p}{k}\cdot opt(G) + d$. 

If $n>p$ then $\mathbf A(G,\pi)\le n-1 = \frac{n-p}{k}\cdot k + (p-1)\le  \frac{n-p}{k}\cdot opt(G) + (p-1) \le  \frac{n-p}{k}\cdot opt(G) +d$. 

\noindent Thus (ii) is verified.

\end{proof}

The above results stands in contrast to the one for the online $\mathsf{MaxCut}$ problem.
In the case of online minimization problems like mincut making a single mistake can prove to be costly.
Can advice help? 
There are two interesting cases to consider.
One where we want to find the optimal cut and the other where an approximate value would suffice.
As it turns out the advice complexity of these two problems are more or less the same.
This is  a bit surprising as there are $\mathsf{AOC}$-complete online problems for which this is not the case.
Here $\mathsf{AOC}$ stands for \emph{asymmetric online cover} which was introduce in \cite{boyar2015advice}.
For problems in this class a $c$-competitive algorithm requires $\Omega(n/c)$-bits of advice and this is tight.
However, the above results for \mc are pessimistic. 
The following two theorems gives the advice complexity for optimality.

\begin{restatable}{theorem}{upnminus}\label{thm:up n-1}
There is an competitive algorithm that finds a minimum cut with $n-1$ bits of advice.
\end{restatable}

\begin{proof}
Let $\mathbf A^{adv}$ define $X_1=\{v_1\}$ and $Y_1=\emptyset$ when it receives $v_1$. For each $i=2,...,n$, suppose $X_{i-1}$ and $Y_{i-1}$ have been constructed. When $v_i$ is revealed $\mathbf A^{adv}$ asks: is $(X_{i-1}\cup \{v_i\},Y_{i-1})$ extendable? If the answer is yes then set $X_i=X_{i-1}\cup \{v_i\}$ and $Y_i=Y_{i-1}$; if the answer is no then set $X_i=X_{i-1}$ and $Y_i=Y_{i-1}\cup \{v_i\}$. At the end, $\mathbf A^{adv}$ finds an optimal solution with $n-1$ bits of advice.
\end{proof}

The algorithm correctly determines a minimum cut even if the graph $G$ is disconnected.
Unfortunately as Theorem \ref{thm:low n-5} shows this naive strategy is almost optimal.

\begin{restatable}{theorem}{lownminus}\label{thm:low n-5}
There is a collection $\mathcal{G}_{(n)}$ of graphs such that any competitive algorithm solving $\mathsf{MinCut}[\mathcal{G}_{(n)}]$ requires at least $n-5$ bits of advice. 
\end{restatable}

\begin{figure}[h]
	\includegraphics[width=6cm]{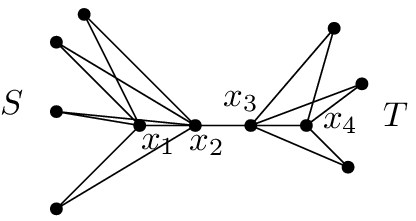}
	\centering
	\caption{The class $\mathcal{G}_{(n)}$ used in the proof of Theorem ~\ref{thm:low n-5} }
\label{fig:opt-adv}
\end{figure}

\begin{proof}
For every $n \ge 6$ we present a graph for which a competitive algorithm requires at least $n-5$ bits of advice.
Each graph $G$ in the collection has path a $P= (x_1,x_2,x_3,x_4)$ of length 4 (see Figure~\ref{fig:opt-adv}). 
Additionally, all other vertices of $G$ are divided into two parts $S$ and $T$.
Each vertex in $S$ is adjacent to both $x_1,x_2$ and each vertex in $T$ is adjacent to both $x_3,x_4$.
There are no other edges in $G$.
Suppose the adversary first reveals the vertices in $S \cup T$.
The induced subgraph  $G[S \cup T]$ forms an independent set.
Let $\Gamma_{n-4}$ be the set of potential graphs that remain after processing the set $\{v_1,v_2,\ldots,v_{n-4}\}$.
First we show $|\Gamma_{n-4}| = 2^{n-4}$.
This follows from the fact that the set $\{v_1,v_2,\ldots,v_{n-4}\}$ can be partitioned in $2^{n-4}$ different ways depending on which vertices (if any) are adjacent to $\{x_1, x_2\}$.
Since the labels $S$ and $T$ are interchangeable there are exactly $|\Gamma_{n-4}|/2 = 2^{n-5}$ pairwise distinct optimal solutions in $\Gamma_{n-4}$. 

An optimal algorithm, without advice, must be able distinguish between these pairwise distinct solutions before the path $P$ is revealed. By the standard information theoretic argument we see that  $\ge n-5$ advice bits are necessary to solve  $\mathsf{MinCut}[\mathcal{G}_{(n)}]$ optimally. 
\end{proof}

Next we ask : how much advice is necessary and sufficient to approximate the value of the mincut value. 
Theorem ~\ref{thm: classic} gives a $O(n/k)$-competitive algorithm even without advice whenever $k \ge 1$.
However, with only $O(\log n + \log \log n)$ bits of advice we can achieve a $\frac{\delta(G)}{k}$-competitive algorithm.
Here $\delta(G)$ is the minimum degree of $G$.
At the beginning we ask the oracle the position of a vertex with the minimum degree, which requires $O(\log n + \log \log n)$ bits.
The $\log \log n$ term correspond to the extra bits used to make the advice string self-delimiting.
The algorithm puts this vertex in one part and all other vertices into the other part resulting in a cut of size $\delta(G)$.
Unfortunately, if  $\delta(G) =  O(n)$ then it is no better than the algorithm without advice.
In the next theorem we show that this is essentially the best one can do.

\begin{restatable}{theorem}{advapprox}\label{thm: adv approx low}
Let $\mathbf A^{adv}$ be a $c$-competitive algorithm for $\mathsf{MinCut}[{\cal G}_{k}]$ where $1 \le k \le \lfloor\frac{n-4}{2}\rfloor$.
For every $\frac{k+1}{n} < \epsilon < \frac{1}{2}(1-\frac{1}{n})$, if
$\mathbf A^{adv}$ uses $b < n - 2\lceil \epsilon n \rceil-1 $ bits of advice 
then  $c \ge  \frac{\epsilon n-1}{k}$.
\end{restatable}

\begin{proof}
We will show that there is an infinite family $\mathcal{G}_{(n)}$ of graphs for which the theorem holds.
Consider a graph $G \in \mathcal{G}_{(n)}$ as shown in Figure~\ref{fig:opt-adv-app}. 
The induced subgraphs $G[A]$ and $G[B]$ are both cliques of size $p > k+1$.
We connect $A$ and $B$ via the sets $A'$ and $B'$.
Since the minimum cut is $k$ we ensure $E(A',B') = k$.
The induced subgraphs $G[C]$ and $G[D]$ are both independent sets and $|C| + |D| \ge 2$.
Each vertex in $C$ (resp. $D$) is adjacent to all vertices in $A$ (resp. $B$).
The adversary sends the vertices in the set $C \cup D$ before sending any of the vertices in $A \cup B$.
Let $\Gamma_{C \cup D}$ be the set of potential graphs 
after $G[C \cup D]$ has been revealed.
Depending on how the vertices in $C \cup D$ are connected to $A \cup B$ there are $2^{|C| + |D|-1} = 2^{n-2p-1}$ pairwise different optimal solutions with a minimum cut of $k$ corresponding to the set $\Gamma_{C \cup D}$.
This is essentially the same argument we used when proving Theorem ~\ref{thm:low n-5}.
With $b$ bits of advice there are only $2^b$ possible advice strings.
Hence there exists some advice string $\phi$ which is read by $\mathbf A^{adv}$ for at least $2^{n-2p-1} / 2^b$ inputs having pairwise different optimal solutions. 
Let this set be $\mathcal S$.
If $|\mathcal S| > 1$ then the adversary can fool $\mathbf A^{adv}$ by choosing an input from $\cal S$ that results in a non-optimal solution when used with $\phi$.
Suppose after reading $\phi$, $\mathbf A^{adv}$ chooses a partition of $C \cup D$ according to a solution $(X',Y')$ (aka a partition of G) in $\mathcal{S}$.
Then adversary sends the rest of $G$ (aka the vertices in $A \cup B$) according to some other partition $(X'',Y'') \in {\mathcal S}$.
Since $\mathbf A^{adv}$ has no means of distinguishing these to case based on the advice string $\phi$ it will fail to optimally partition $G$.
It is easy to see that  for any non-optimal partition $(X',Y') \ne (X,Y)$ of $G$ we have $\lambda(X',Y') \ge p-1 = \frac{p-1}{k}opt(G)$.
Thus we must have $|\mathcal{S}| \le 1$, which implies $n-2p-1-b \le 0$.
Taking $p = \lceil \epsilon n \rceil$ we see $b \ge n - 2\lceil \epsilon n \rceil-1 $ if $\mathbf A^{adv}$ to be less than $c$-competitive. 
\end{proof}


Theorems \ref{thm:up n-1}-\ref{thm: adv approx low} together shows a limitation of the advice  model. Unlike  $\mathsf{AOC}$-complete problems the advice complexity for \mc has a  sharp phase transition. Either we have sufficient amount of advice to produce an optimal solution or a sub-linear competitive ratio cannot be guaranteed. 


\subsection{Semi-adversarial Models: Random Vertex Order}
In the previous section we showed that there is a $O(n/k)$-competitive algorithm when both the input graph and the order of arrival is determined by an adversary.
This upper bound also holds when the order of arrival is determined by a random permutation.
Unfortunately, it turn's out this is the best we can do without any restriction on the input graph.
We show this next. We complement this lower bound result with an $O(1)$ upper bound for sparse connected graphs.

\begin{restatable}{theorem}{randorderlow}\label{thm: rorder}
For any deterministic algorithm $\mathbf A$ for   $\mathsf{MinCut}[{\cal G}_{k}]$  under the random-vertex order model there exists a class of graphs $\mathcal{G}_{(n)}$ for infinitely many values of $n$ for which,  $$\mathbb{E}[\mathbf{A}(G)] \ge \frac{n}{64k} \cdot opt(G).$$ 
 Here the expectation is taken over the random order.
\end{restatable}

\begin{figure}[h]
	\includegraphics[width=9cm]{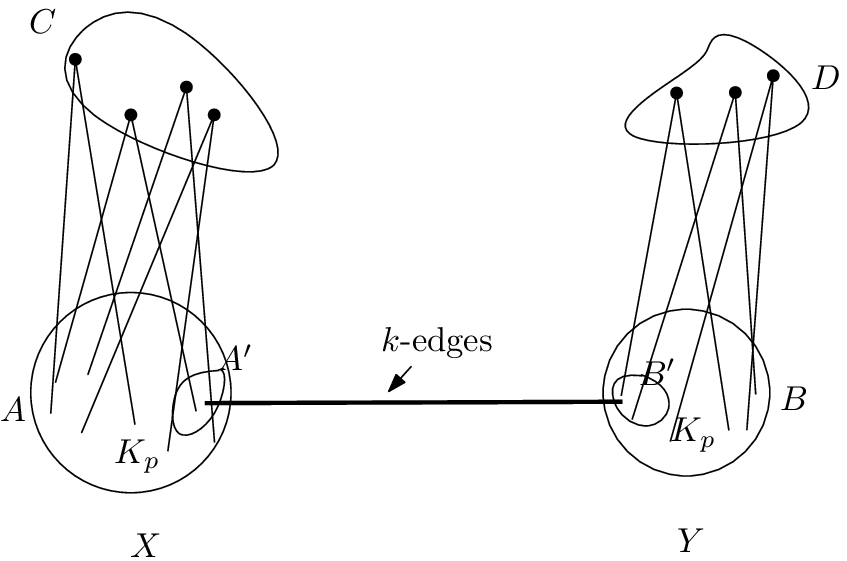}
	\centering
	\caption{A graph $G \in \mathcal{G}_{(n)}$ used in the proof of Theorem ~\ref{thm: adv approx low}}
\label{fig:opt-adv-app}
\end{figure}

\begin{proof}
We use the  class of graphs $\mathcal{G}_{(n)}$ from Theorem ~\ref{thm: adv approx low} (Figure~\ref{fig:opt-adv-app}).
Here we take $|C|=|D|=\epsilon n$ and $|A|=|B| = (1/2-\epsilon)n$.
We note that $\lambda(G) = k$, same as before.
An optimal offline algorithm returns this value.
Consider any online algorithm $\mathbf{A}$. 
Without loss of generality we may assume $v_1$ is assigned to the partition $X$.
Let $V_i = \{v_1,\ldots,v_i\}$ be the set of vertices to arrive so far.
Let $\mathbb{E}[\mathbf{A}(G, V_i)]$ be the
expected value of the mincut computed by the online algorithm after processing the vertices $v_1$ through $v_i$.
Let $\mathbf{A}_i$ be the following algorithm which has two phases: online and offline.
In the online phase it processes the first $i$ vertices same as $\mathbf{A}$ creating a partial solution.
Then it is allowed to read the rest of the input just like an offline algorithm.
This is the offline phase.
It  outputs a final partition that minimizes the cut value while respecting the decisions made during its online phase.
Let $\mathbb{E}[\mathbf{A}_i(G)]$ be the expected value of the minimum cut computed by $\mathbf{A}_i$.
It is clear that $\mathbb{E}[\mathbf{A}(G)]=\mathbb{E}[\mathbf{A}(G, V_n)]= \mathbb{E}[\mathbf{A}_n(G)] $.
Further, the function $\mathbb{E}[\mathbf{A}_i(G)]$ is monotonically increasing in $i$.
Hence we have, $$\mathbb{E}[\mathbf{A}(G, V_n)] \ge \mathbb{E}[\mathbf{A}_2(G)]$$ 
We give a lower bound for $\mathbb{E}[\mathbf{A}_2(G)]$ as claimed in the theorem.
Let $\lambda(G,X_i,Y_i)$ be the minimum cut achievable after assigning the first $i$ vertices by $\mathbb{A}$, where $(X_i, Y_i)$ is the resulting partition.
There are two cases as follows.

\paragraph{Case 1:}
 [$v_1$ and $v_2$ are not adjacent]. $\mathbf{A}_2$ either puts (1) both of them in $X$ or (2) puts $v_2$ in $Y$. Suppose $\mathbf{A}_2$ chooses (1). Then,
    \begin{align}\label{eq: exp1}
    \mathbb{E}&[\mathbf{A}_2(G) |\ v_1,v_2\ \mbox{are not adjacent}]  =
    \mathbb{P}[v_1,v_2 \in C\ \mbox{or}\  v_1,v_2 \in D]\cdot k\nonumber \\& 
    + 2\mathbb{P}[v_1 \in A\ \mbox{and}\  v_2 \in B]\cdot \alpha_1 \nonumber\\& 
    + 2\mathbb{P}[v_1 \in C\ \mbox{and}\  v_2 \in D]\cdot \alpha_2 \nonumber \\
    & \ge 2\mathbb{P}[v_1 \in C\ \mbox{and}\  v_2 \in D]\cdot \alpha_2
    \end{align}
Here $\alpha_2$ is a lower bound on the minimum cut found by  $\mathbf{A_2}$ when $v_1$ and $v_2$ are in different stable sets $C$ and $D$. Similarly we define $\alpha_1$ (which is ignored).
Clearly $\alpha_2 \ge |A| = (1/2-\epsilon)n$.
Since $v_1, v_2$ are picked from a random order,

\begin{align}
  &  \mathbb{P}[v_1 \in C\ \mbox{and}\  v_2 \in D] = \epsilon^2
\end{align}

\noindent From Equation \ref{eq: exp1} we get:

\begin{align}
    \mathbb{E}&[\mathbf{A}_2(G) |\ v_1,v_2\ \mbox{are not adjacent}] \ge  2\epsilon^2(1/2-\epsilon)n 
\end{align}

\noindent Now suppose $\mathbf{A}_2$ puts $v_2$ in $Y$.
A similar argument to the one above can be made to show that,
\begin{align}
    \mathbb{E}&[\mathbf{A}_2(G) |\ v_1,v_2\ \mbox{are not adjacent}] \ge  2 \epsilon^2(1/2-\epsilon)n
\end{align}

\paragraph{Case 2:}[$v_1$ and $v_2$ are adjacent.] Again we have two possibilities. (1) $\mathbf{A}_2$ puts $v_2$ in $X$ and (2) $\mathbf{A}_2$ puts $v_2$ in $Y$.
For the first case we have,

   \begin{align}\label{eq: exp}
    \mathbb{E}&[\mathbf{A}_2(G) |\ v_1,v_2\ \mbox{are  adjacent}]  =
    \mathbb{P}[v_1,v_2 \in A\ \mbox{or}\  v_1,v_2 \in B]\cdot k\nonumber \\& 
    + 2\mathbb{P}[v_1 \in A\ \mbox{and}\  v_2 \in B]\cdot (1/2-\epsilon)n \\&
    \ge (1/2-\epsilon)^3n
    \end{align}

\noindent In a similar manner we find that if $\mathbf{A}_2$ puts $v_2$ in $Y$ then,

   \begin{align}\label{eq: exp2}
    \mathbb{E}&[\mathbf{A}_2(G) |\ v_1,v_2\ \mbox{are  adjacent}]  \ge \epsilon^2(1/2-\epsilon)n
    \end{align}
In all of the of the above cases  regardless of what $\mathbf{A}_2$ chooses do with $v_2$ we have ,
$$\mathbb{E}[\mathbf{A}_2(G)] \ge \min (\epsilon^2(1/2-\epsilon)n,(1/2-\epsilon)^3n)$$
The right hand side of the above expression is maximized when $\epsilon = 1/4$ and we get $\mathbb{E}[\mathbf{A}_2(G)] \ge n/64$.
\end{proof}

\subsection{Semi-adversarial Models: Specific Graph Classes}
\subsubsection{Sparse Connected Graphs}

In this section we present a result on sparse connected graphs. Sparseness here is defined to mean that the graph has linear number of edges.
Many important families of graphs falls in this category such as planer graphs, degree bounded expanders etc.

\begin{algorithm}[h!]
\caption{ An algorithm for sparse graphs}
\begin{algorithmic}[1]
\STATE {\bf Input:} A sparse connected graph $G$ with an vertex arrival order chosen uniformly at random.
\STATE {\bf Output:} A cut of $G$.
\STATE Initialize: $X \leftarrow \emptyset$ , $Y \leftarrow \emptyset$ and $i \leftarrow 1$.
\WHILE{$i \le n$}
\IF{$i == 1$}
\STATE $Y \leftarrow \{v_i\}$
\ELSE 
\STATE $X \leftarrow X \cup \{v_i\}$
\ENDIF 
\STATE $i \leftarrow i + 1$
\ENDWHILE
\end{algorithmic}
\label{alg: sparse}
\end{algorithm}

\begin{restatable}{theorem}{sparsrand}\label{thm: sparserand}
In the random order model there is $O(1)$-competitive algorithm in expectation for sparse connected graphs with $O(n)$ edges.
\end{restatable}

\begin{proof}
We show that Algorithm ~\ref{alg: sparse} is $O(1)$-competitive.
Suppose the graph $G$ has $O(n)$ edges and is connected.
Algorithm ~\ref{alg: sparse} essentially  puts a random vertex in $Y$ and rest in $X$.
Since $G$ is connected $\lambda(G) \ge 1$.
Let $v^*$ be the vertex chosen to be in $Y$ and $\mathbb{E}[d(v^*)]$ be its expected degree.
Let $d_1 > \ldots > d_n$ be the degree sequence of $G$.
The number of vertices of degree $d_i$ is at most $n_i$.
From the first theorem of graph theory we have $\sum_i n_id_i = 2|E| = O(n)$.

\noindent Now,
$$\mathbb{E}[d(v^*)] = \sum_i \mathbb{P}[d(v^*) = d_i]d_i = \frac{1}{n}\sum_in_id_i = \frac{1}{n}O(n) = O(1)$$

\noindent Hence the competitive ratio is bounded.

\end{proof}
\begin{corollary}
For a class of connected graphs with $m$-edges there is an $O(\frac{m}{n})$-competitive algorithm.
\end{corollary}
\begin{proof}
Immediately follows from theorem \ref{thm: sparserand}. 
\end{proof}



\subsubsection{Dense Random Graphs ($p = \omega(\log n/n)$)}
In this model the input graph itself is random. We denote by $G_{n,p}$ as the random graph generated according to the Erdos-Renyi model \cite{bollobas1985random}. That is, each edge of $G_{n,p}$ is present with probability $p$ which is independent of other edges.
We take $q =  1-p$.
It is well known that for such graphs the expected minimum degree $\mathbb{E}[\delta(G_{n,p})]$ and the expected size of the minimum cut $\mathbb{E}[\lambda(G_{n,p})]$ are closely related. Specifically, in \cite{bollobas1985random} it was shown that:
\begin{align}
    \mathbb{P}[\lambda(G_{n,p})=\delta(G_{n,p})] \to 1\ \mbox{as $n \to \infty$} 
\end{align}
for almost every $G_{n,p}$ with high probability\footnote{The probability tends to 1 as $n \to \infty$ } (w.h.p.). Hence for a random graph, approximating the minimum cut is equivalent to approximating the minimum degree  w.h.p.
In the case of dense graphs we assume $p$ is sufficiently large so the $G_{n,p}$ is connected w.h.p.
It is well known that taking $p = \omega(\log n / n)$  suffice for this purpose \cite{bollobas2001random}.

We discussed the vertex random order model in the previous section.
In that model, the above problem seems similar to the well known \emph{Secretary problem}.
However there are some important differences.
Recall that in the classical Secretary problem\footnote{Here we use the adjective classical to differentiate it from several of its variants which were developed subsequently.} there is a set $S$ of $n$ secretaries. 
Secretaries are ranked according to some total order.
In the online problem, a random order is selected and secretaries arrive one at a time for their interview.
The algorithm must decide whether to hire the secretary or reject them.
Both of these are irrevocable decisions.
The problem is to come up with a strategy that maximizes the probability of hiring the best secretary.
The well known optimal solution is to reject the first $\lfloor n/e\rfloor$ secretaries (here $e$ is the base of the natural logarithm) and accept among the subsequent candidates the first secretary whose rank is better than the secretaries interviewed so far.
It was shown in \cite{chow1964optimal} that it is possible to obtain a rank in expectation which is about $3.87$ times that of the optimal rank 1.
Thus under the random order model the classical secretary problem has a constant competitive ratio.

\subsubsection{Approximating  $\delta(G_{n,p})$ }
First we discuss  why  the minimum degree estimation using  the secretary selection strategy  fails for a $G$ which is adversarially generated.
\ifx false
\begin{align}
    \mathbb{E}[\mathbb{A}(\delta(G_{n,p}))]\le \mathbb{E}[\mathbb{A}(\delta(G))]
\end{align}

where the expectation on the left hand side is taken over the graph random variable $G_{n,p}$ and the right hand side expectation is taken over the random ordering of the vertices.
\fi
The upper bound of the rank in \cite{chow1964optimal} is not enough in this case.
Let $d_1 \ge d_2 \ge \ldots \ge d_n$ be the degree sequence of $G$.
It is possible to have a degree sequence of the following from:
$$d_1 = n-1=d_2=\ldots=d_{\delta(G)}, d_{\delta(G)+1}=n-2=\ldots=d_{n-1}\ \mbox{and}\  d_{n}=\delta(G).$$
Let the expected rank of the selected vertex using the above secretary selection strategy  be $m^*$.
If $G$ is adversarially generated, then $d_{m^*} \ge cn$  for some constant $c > 0$.
However, in our case we are dealing with random graphs for which we can avoid this situation.

\ifx false
\begin{theorem}{(Theorem 3.16 in \cite{bollobas2001random})}
Suppose $m(n) \to \infty $ when $n \to \infty$. Then for almost every random graph $G_{n,p}$ the following holds:
\begin{align}\label{eq: deg-diff}
    d_i - d_{i+1} \le \frac{\omega(1)}{(m(n))^2}\left(\frac{pqn}{\log n}\right)^{1/2} \ \mbox{and}\ i \le m
\end{align}
where $\omega(1)$ is a function on $n$ that grows arbitrarily slowly and $m/n < 1$.

\end{theorem}
\fi
Our algorithm is simple, in fact trivial.
We assume no knowledge of the size of the input graph. 
This makes sense when proving an upper bound as is the case here.
Since, we must create a valid partition we simply put the first vertex to arrive in the part $X$ and rest to part $Y$.
In this case the expected  value of the cut will be the same as the expected degree of the first vertex which is  $np$.
It is easy to see that this is the best one can do.

\begin{theorem}
The above strategy approximates the min-cut within a constant factor w.h.p. whenever $p = \omega(\log n / n)$.
\end{theorem}

\begin{proof}
It is well known that if $\frac{np}{\log n} \to \infty$ as $n \to \infty$ then  w.h.p. $\delta(G_{n,p}) \ge \epsilon np$ for some constant $\epsilon > 0$. See for example Theorem 3.4 in \cite{frieze2016introduction}.
\ifx false
\begin{itemize}
    \item[Case 1:] When $G_{n,p}$ is sparse. Here we assume $\frac{np}{\log n} \to \infty$ as $n \to \infty$. It is known that [ref] in this case for almost every $G_{n,p}$, w.h.p. $\delta(G_{n,p}) \ge \epsilon np$ for some constant $\epsilon > 0$.
    \item[Case 2:] When $G$ is dense. Here we assume $0 < p < 1$ is a constant. For this case it is known that, w.h.p.,
    \begin{align}
        \delta(G_{n,p}) \ge n(1-p) - o(n)
    \end{align}
\end{itemize}
\fi
Since the expected degree of a random vertex is Binomially distributed with mean $np$ using Chernoff bounds we can show that $\mathbf{A}(G_{n,p}) \le \delta np$ for some constant $\delta > 0$, w.h.p.
\end{proof}

\noindent Note that when $p = O(\log n / n)$ then the graph has an isolated vertex w.h.p.


\ifx false
\section{Adversarial Input and Advice Complexity}

\noindent Next we investigate if having access to an advice oracle is helpful for this problem.


\subsection{Advice Complexity of $\mathsf{MinCut}$}

 In theorem ~\ref{thm:up n-1} and theorem ~\ref{thm:low n-5} we give upper and lower bounds for the advice complexity of competitive algorithms.

\begin{theorem} \label{thm:up n-1}
There is an competitive algorithm that finds a minimum cut with $n-1$ bits of advice.
\end{theorem}
\begin{proof}
Let $\mathbf A^{adv}$ define $X_1=\{v_1\}$ and $Y_1=\emptyset$ when it receives $v_1$. For each $i=2,...,n$, suppose $X_{i-1}$ and $Y_{i-1}$ have been constructed. When $v_i$ is revealed $\mathbf A^{adv}$ asks: is $(X_{i-1}\cup \{v_i\},Y_{i-1})$ extendable? If the answer is yes then set $X_i=X_{i-1}\cup \{v_i\}$ and $Y_i=Y_{i-1}$; if the answer is no then set $X_i=X_{i-1}$ and $Y_i=Y_{i-1}\cup \{v_i\}$. At the end, $\mathbf A^{adv}$ finds an optimal solution with $n-1$ bits of advice.
\end{proof}

Note that in this case answer to each question is a 1-bit answer.
Hence there are no overhead due to self-delimiting strings.
The algorithm correctly determines a minimum cut even if the graph $G$ is disconnected.

\begin{figure}[h]
	\includegraphics[width=6cm]{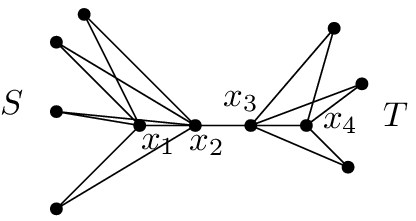}
	\centering
	\caption{The class $\mathcal{G}_{(n)}$ used in the proof of theorem ~\ref{thm:low n-5} }
\label{fig:opt-adv}
\end{figure}

\begin{theorem}\label{thm:low n-5}
There is a collection $\mathcal{G}_{(n)}$ of graphs such that any competitive algorithm solving $\mathsf{MinCut}[\mathcal{G}_{(n)}]$ requires at least $n-5$ bits of advice.
\end{theorem}

\begin{proof}
For every $n \ge 6$ we present a graph for which a competitive algorithm requires at least $n-5$ bits of advice.
Each graph $G$ in the collection has path a $P= (x_1,x_2,x_3,x_4)$ of length 4 (see Figure~\ref{fig:opt-adv}). 
Additionally, all other vertices of $G$ are divided into two parts $S$ and $T$.
Each vertex in $S$ is adjacent to both $x_1,x_2$ and each vertex in $T$ is adjacent to both $x_3,x_4$.
There are no other edges in $G$.
Suppose the adversary first reveals the vertices in $S \cup T$.
The induced subgraph graph $G[S \cup T]$ forms an independent set.
Let $\Gamma_{n-4}$ be the set of potential graphs remain after processing the set $\{v_1,v_2,\ldots,v_{n-4}\}$.
First we show $|\Gamma_{n-4}| = 2^{n-4}$.
This follows from the fact that the set $\{v_1,v_2,\ldots,v_{n-4}\}$ can be partitioned in $2^{n-4}$ different ways depending on which vertices (if any) are adjacent to $\{x_1, x_2\}$.
Since the labels $S$ and $T$ are interchangeable there are exactly $|\Gamma_{n-4}|/2 = 2^{n-5}$ pairwise distinct optimal solutions in $\Gamma_{n-4}$. 

An optimal algorithm, without advice, must be able distinguish between these pairwise distinct solutions before the path $P$ is revealed. By the standard information theoretic argument we see that  $\ge n-5$ advice bits are necessary to solve  $\mathsf{MinCut}[\mathcal{G}_{(n)}]$ optimally. 
\end{proof}

Next we ask the question how much advice is necessary and sufficient to approximate the value of the mincut value. 
Theorem ~\ref{thm: classic} gives a $O(n/k)$-competitive algorithm even without advice whenever $k \ge 1$.
However, with only $O(\log n + \log \log n)$ bits of advice we can achieve a $\delta(G) / k$-competitive algorithm.
Here $\delta(G)$ is the minimum degree of $G$.
At the beginning we ask the oracle the position of a vertex with the minimum degree, which requires $O(\log n + \log \log n)$.
The $\log \log n$ term correspond to the extra bits used to make the advice string self-delimiting.
The algorithm puts this vertex in one part and all other vertices into the other part resulting in a cut of size $\delta(G)$.
Unfortunately, if  $\delta(G) =  O(n)$ then it is no better than the algorithm without advice.
In the next theorem we show that this is essentially the best one can do.

\begin{theorem}\label{thm: adv approx low}
Let $\mathbf A^{adv}$ be a $c$-competitive algorithm for $\mathsf{MinCut}[{\cal G}_{k}]$ where $1 \le k \le \lfloor\frac{n-4}{2}\rfloor$.
For every $\frac{k+1}{n} < \epsilon < \frac{1}{2}(1-\frac{1}{n})$, if
$\mathbf A^{adv}$ uses $b < n - 2\lceil \epsilon n \rceil-1 $ bits of advice 
then  $c \ge  \frac{\epsilon n-1}{k}$.
\end{theorem}

\begin{figure}[h]
	\includegraphics[width=7cm]{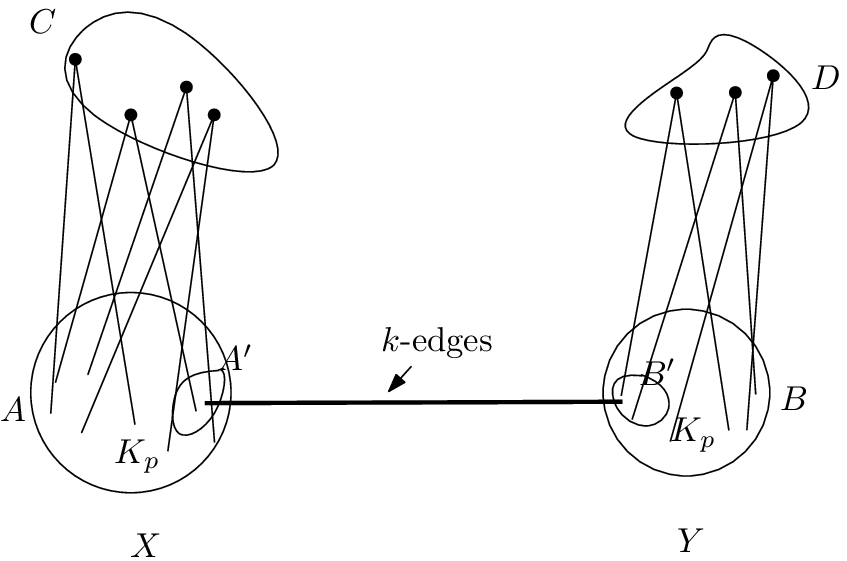}
	\centering
	\caption{A graph $G \in \mathcal{G}_{(n)}$ used in the proof of theorem ~\ref{thm: adv approx low}}
\label{fig:opt-adv-app}
\end{figure}

\begin{proof}
We will show that there is an infinite family $\mathcal{G}_{(n)}$ of graphs for which the theorem holds.
Consider a graph $G \in \mathcal{G}_{(n)}$ as shown in Figure~\ref{fig:opt-adv-app}. 
The induced subgraphs $G[A]$ and $G[B]$ are both cliques of size $p > k+1$.
We connect $A$ and $B$ via the sets $A'$ and $B'$.
Since the minimum cut is $k$ we ensure $E(A',B') = k$.
The induced subgraphs $G[C]$ and $G[D]$ are both independent sets and $|C| + |D| \ge 2$.
Each vertex in $C$ (resp. $D$) is adjacent to all vertices in $A$ (resp. $B$).
The adversary sends the vertices in the set $C \cup D$ before sending any of the vertices in $A \cup B$.
Let $\Gamma_{C \cup D}$ be the set of potential graphs 
after $G[C \cup D]$ has been revealed.
Depending on how the vertices in $C \cup D$ are connected to $A \cup B$ there are $2^{|C| + |D|-1} = 2^{n-2p-1}$ pairwise different optimal solutions with a minimum cut of $k$ corresponding to the set $\Gamma_{C \cup D}$.
This is essentially the same argument we used when proving theorem ~\ref{thm:low n-5}.
With $b$ bits of advice there are only $2^b$ possible advice strings.
Hence there exists some advice string $\phi$ which is read by $\mathbf A^{adv}$ for at least $2^{n-2p-1} / 2^b$ pairwise different optimal solutions. 
Let this set be $\mathcal S$.
If $|\mathcal S| > 1$ then the adversary can fool $\mathbf A^{adv}$ in choosing a non-optimal solution.
Suppose after reading $\phi$, $\mathbf A^{adv}$ chooses a partition of $C \cup D$ according to a solution $(X',Y')$ (aka a partition of G) in $\mathcal{S}$.
Then adversary sends the rest of $G$ (aka the vertices in $A \cup B$) according to some other partition $(X'',Y'') \in {\mathcal S}$.
Since $\mathbf A^{adv}$ has no means of distinguishing these to case based on the advice string $\phi$ it will fail to optimally partition $G$.
It is easy to see that  for any non-optimal partition $(X',Y') \ne (X,Y)$ of $G$ we have $\lambda(X',Y') \ge p-1 = \frac{p-1}{k}opt(G)$.
Thus we must have $|\mathcal{S}| \le 1$, which implies $n-2p-1-b \le 0$.
Taking $p = \lceil \epsilon n \rceil$ we see $b \ge n - 2\lceil \epsilon n \rceil-1 $ if $\mathbf A^{adv}$ to be less than $c$-competitive. 
\end{proof}

In some respect, the minimum cut problem captures a limitation of the advice  model. 
Unless we have sufficient amount of advice (say $b$) to produce an optimal solution, having any amount of advice significantly less that that, say $\le \epsilon b$ is essentially useless and a similar competitiveness can be achieved without any advice.
 Results in this section can be  extended to the weighted version be multiplying the factor $\frac{w_{max}}{w_{min}}$ in the competitiveness term.
\fi 

\section{MinCut Under Non-stationary Regret}
In Section 2 we see that, apart from some special cases,   \mc does not exhibit competitive algorithms.
So we focus our attention to maintaining a cut as close to the minimum as possible under a (discrete) time varying weight function.
We re-state our problem in the framework of regret analysis where we consider the non-stationary case. To the best of our knowledge this has not been done before.

At the beginning we are given $G_0(V, E)$ and a dummy weight function $w_0: E \to \mathbb R^+$. 
At time step $t$ the online algorithm must chooses a cut based on the knowledge of the weight functions observed thus far.
Once the algorithm plays a cut the adversary reveals $w_t$.
The goal of the algorithm is to minimize the non-stationary regret as defined in section 1.2. 

Note that $w_t$'s are linear with cut constraints and $\mathbf{A}$ has full access to the preceding sequence.
However, without any variational bound the regret cannot be sublinear.
For example consider the graph $P_3$ with two edges $e_1, e_2$.
Consider the two weight functions $w_A, w_B$. 
Let $w_A(e_1) =  1$ and $w_B(e_2) = 0$.  Let $w_B(e_i) = 1 - w_A(e_i)$, $i \in \{1,2\}$.
At time step $t$ the adversary chooses either $w_A$ or $w_B$ with probability 0.5.
Clearly, regardless of what strategy $\mathbf{A}$ plays the expected value of the cut will be at least 0.5. However the optimal cut value is 0 for every $t$.
Hence the regret increases linearly with $T$ in expectation.
Although this is crucial technical reason for us to restrict the variation of the weight functions it make sense in practice as well.
For example if the graph models a communication network it is reasonable to assume that the overall changes to the network traffic is bounded even if some edges may experience significant fluctuations in their traffic during certain periods. 
We model this by assuming that the total variation of the weights are bounded in the following way:
\begin{align}\label{eq: var}
    \mathcal{F}_T = \{(w_1,\ldots,w_T)\mid \sum_{t=1}^{T-1}{||w_t - w_{t+1}||_1 \le V_T}\}
\end{align}
where $||\cdot||_1$ is the Manhattan distance between the successive weight functions. 
Note the distinction between this and that of Equation \ref{eq: variation}.
Here the variational budget is not directly specified in terms of the feedback function, which  gives the value of the cut corresponding to a  set of feasible edges.
This makes sense in the context of graphs.
The feedback (loss function) is the cut function $C_t(X) = w_t(X, V\setminus X)$, the weight of the edges crossing the two parts $X$ and $Y = V\setminus X$.
The regret function is given by:
\begin{align}\label{eq:rand non  reg}
    \eregn(\mathbf{A}) = \mathbb{E}\left[\sum_{t=1}^T{C_t(X_t)}- \sum_{t=1}^T {C_t(X_t^*)}\right]
\end{align}

\begin{algorithm}[h!]
\caption{ The follow the current optimal ($\mathsf{FTCO}$) algorithm}
\begin{algorithmic}[1]
\STATE {\bf Input:} A graph $G(V, E)$. At time step $t$ a weight function $w_t$ satisfying Equation \ref{eq: var}.
\STATE {\bf Output:} A sequence of cuts of $G$.
\STATE Initialize a dummy weight function $w_0 : E \to \{1\}$.
\STATE $t \leftarrow 1$
\COMMENT{Let $G_t$ be the graph corresponding to the weight function $w_t$}
\WHILE{there is a new weight function}
\STATE Find a minimum cuts of $G_{t-1}$. Let this correspond to the partition $(X_t, Y_t)$
\STATE Play this cut at step $t$.
\STATE $t \leftarrow t+1$
\ENDWHILE
\end{algorithmic}
\label{alg: FTCO}
\end{algorithm}

\begin{theorem}{(Upper bound)}
$\mathsf{FTCO}$ has a regret of $O(V_T)$.
\end{theorem}
\begin{proof}
We need to bound the right hand side of the expression in Equation \ref{eq:rand non  reg}. Since $\mathsf{FTCO}$ plays a minimum cut of the previous step we have:
\begin{align*}
 C_t(X_{t}) = C_{t}(X_{t-1}^*) = C_{t-1}(X_{t-1}^*) + \Delta C_t  
\end{align*}
Where $\Delta C_t$ is the change in the weight of the cut due to change in the weight function going from step $t-1$ to the step $t$.
Using this expression for $C_t(X_t)$ in Equation \ref{eq:rand non  reg} we see that:
\begin{align*}
    \eregn(\mathbf{A}) = & \sum_{t=1}^T{C_{t-1}(X_{t-1}^*) + \Delta C_t}- \sum_{t=1}^T {C_t(X_t^*)} = \sum_{t=1}^T{C_{t-1}(X_{t-1}^*)-C_t(X_t^*)} + \sum_{t=1}^T\Delta C_t\\
    = & C_0(X_0^*) - C_T(X_T^*) + \sum_{t=1}^T\Delta C_t
\end{align*}
Since the quantity $C_0(X_0^*) - C_T(X_T^*) \le O(n)$ is independent of $T$ we only need to bound the summation over $\Delta C_t$'s.
This can be done easily:
\begin{align*}
    \sum_{t=1}^T\Delta C_t = \sum_{t=1}^T\sum_{e \in E(X_{t-1}^*, Y_{t-1}^*)}{|w_{t-1}(e) - w_t(e)|} \le \sum_{t=1}^T{||w_{t-1}-w_{t}||_1} \le V_T 
\end{align*}
This proves the claim of the theorem.
\end{proof}
Note that our algorithm is deterministic.
The following lower bound shows that $\mathsf{FTCO}$ is in fact optimal in our model even if randomization is allowed.
\begin{theorem}{(Lower bound)}
Under the above model for any randomized algorithm $\mathbf{A}$ there is graph and a sequence of weight functions such that $\eregn(\mathbf{A}) = \Omega(V_T)$
\end{theorem}

\begin{proof}
Consider a path $P_n$. At time $t$ the adversary picks one edge uniformly at random and assigns it the weight $1-\epsilon_t \in [0, 1]$. All other edges have weight 1.
So that the variational budget is 
\begin{align}\label{eq: vt}
V_T = \sum_{t=1}^{T-1}{||w_t-w_{t+1}||_1}  \le \sum_{t = 1}^{T-1}{\epsilon_{t}+ \epsilon_{t+1}} \le 2\sum_{t=1}^T\epsilon_t
\end{align}

Hence at time $t$ the minimum cut $C_t(X^*_t) = 1-\epsilon_t$ and for all $X \not \in \argmin_{X \in 2^{V}-\emptyset} C_t(X)$, $C_t(X) \ge 1$.
Now consider the situation encountered by our algorithm $\mathbf{A}$.
We assume  $\mathbf{A}$ and the adversary know each others (mix) strategy.
Even with this information and knowing in advance the structure of the graph the best $\mathbf{A}$ can do is pick an edge uniformly at random and play it as the cut edge.
Choice of any other distribution will only make the adversaries job easier. 
For example, in one extreme case if $\mathbf{A}$ always picks one particular edge then adversary can keep the weight of that edge $\ge 1+c$, for any constant $c> 0$.
Thus we conclude that, $\eregn(\mathbf{A})$ is minimized if $\mathbf{A}$ plays a cut as described above.
In this case we have,

\begin{align*}\label{eq:rand non  reg}
    \eregn(\mathbf{A}) = & \mathbb{E}\left[\sum_{t=1}^T{C_t(X_t)}- \sum_{t=1}^T {C_t(X_t^*)}\right] = \sum_{t=1}^T{\mathbb{E}[C_t(X_t)]} - \sum_{t=1}^T(1-\epsilon_t)\\
    \ge & \sum_{t=1}^T\left(\frac{1}{n^2}(1-\epsilon_t)+\left(1-\frac{1}{n^2}\right)\cdot 1\right) - \sum_{t=1}^T(1-\epsilon_t)\\
    = & \left(1-\frac{1}{n^2}\right)\sum_{t=1}^T\epsilon_t \ge \left(1-\frac{1}{n^2}\right) \frac{V_T}{2}
\end{align*}
Where the first inequality follows from the fact that uniform distribution minimizes the expected cut. And the last one by substituting $V_T$ from Equation \ref{eq: vt}.
\end{proof}


\section{Greedy property of online MinCut and MaxCut}
As we have discussed in section 1.1.3, the performance of $\min_\pi \mathbf A(G,\pi)$ could serve as a measure on the complexity of an online problem $\mathsf P$. In this section we will study $\min_\pi \mathbf A(G,\pi)$ for online mincut and maxcut problems. In both cases, we establish that there exists $\mathbf A$ satisfying $\min_\pi \mathbf A(G,\pi)=opt(G)$ for all $G$. In addition, we obtain an analogous result for maximizing a submodular function and we refute the existence of such a result for minimizing a submodular function. 

  For the current discussion, we allow parallel edges but not loops in $G$. This is the same as allowing a nonnegative weight $w$ on edges and measuring the size of a cut $C$ by the total weight $\sum\{w(e): e\in C\}$.
For any disjoint $X,Y\subseteq V$, let $|X,Y|$ denote the number of edges of $G$ between $X$ and $Y$. We will write $|x,Y|$ or $|X,y|$ for $|X,Y|$ if $X=\{x\}$ or $Y=\{y\}$, respectively. If $E(X,Y)$ is a minimum or maximum cut of $G$ for a partition $(X,Y)$ of $V$ then we may simply call $(X,Y)$ is a minimum, maximum cut of $G$, respectively. If $U\subseteq V$ then we use $G[U]$ to denote the subgraph of $G$ induced on $U$.
 
We will consider a greedy type algorithm $\mathbf A$. Let $\pi =v_1, v_2, ..., v_n$ be a permutation of $V$. Let $X,Y$ be the partition determined by $\mathbf A$ during the process. Since $n$ is unknown to the algorithm, $\mathbf A$ has to place $v_1\in X$ and $v_2\in Y$, because $\mathbf A$ needs to ensure $X\ne\emptyset\ne Y$ even when $n=2$. In the $i$th iteration ($i\ge3$), vertex $v_i$ is revealed and $\mathbf A$ need to decide if $v_i$ should go to $X$ or $Y$. A simple greedy strategy is to make the choice depending on $f_X$ and $f_Y$, which are the number of edges from $v_i$ to $X$ and $Y$, respectively. In the mincut problem, $v_i$ goes to $X$ if $f_X>f_Y$, while in the maxcut problem, $v_i$ goes to $X$ if $f_X<f_Y$. When $f_X=f_Y$, $\mathbf A$ needs to have a tie breaking rule to decide where $v_i$ should go. 

Such a greedy strategy is a common sense approach. The difficulty in studying such an algorithm is to come up with a simple tie breaking rule. It turns out that letting $v_i$ go with $v_{i-1}$ will make things work. To be more specific, in case $f_X=f_Y$, then $v_i$ goes to $X$ if $v_{i-1}$ went to $X$, and $v_i$ goes to $Y$ if $v_{i-1}$ went to $Y$. Let $\mathbf A_{\min}$ and $\mathbf A_{\max}$ be our greedy algorithms with this tie breading rule for mincut and maxcut problems, respectively. For every graph $G$ with two or more vertices, we constructed two permutations $\pi_*$ and $\pi^*$ such that $\mathbf A_{\min}(G,\pi_*)$ is a minimum cut of $G$, and $\mathbf A_{\max}(G,\pi^*)$ is a maximum cut of $G$. To achieve this, we need the following graph theoretical result.

\begin{restatable}{theorem}{gmincut}\label{lem:mincut}
Every loopless graph $G=(V,E)$ has a minimum cut $(X,Y)$ for which there exists a permutation $v_1, ..., v_n$ of $V$ such that the following conditions are satisfied. For each $i\ge1$, let $X_i=X\cap\{v_1,...,v_i\}$ and $Y_i=Y\cap \{v_1,..., v_i\}$. \\ 
(i) $v_1\in X$ and $v_2\in Y$. \\ 
(ii) For every $i\ge3$, if $v_i\in X$ then $|v_i,X_{i-1}|\ge |v_i,Y_{i-1}|$ and if $v_i\in Y$ then $|v_i,Y_{i-1}|\ge |v_i,X_{i-1}|$. \\ 
(iii) If $i\ge3$ is minimum with $v_i\in X$ then $i=|Y|+2$ and $|v_i,X_{i-1}|>|v_i,Y_{i-1}|$. 
\end{restatable}

\begin{proof}
Let us choose a minimum cut $(X,Y)$ with $|X|$ as small as possible. We prove that, with respect to this partition $(X,Y)$, there exists a permutation satisfying (i-iii). 

\noindent Claim 1. If $|X|=1$ then the desired permutation exists.

Let $v_1$ be the unique member of $X$ and let $v_2$ be an arbitrary vertex of $Y$. We prove that there is a desired permutation starting with the two specified terms $v_1,v_2$. Note that no matter how the permutation $v_3,...,v_n$ is determined, conditions (i) and (iii) are always satisfied. So when we define $v_3,...,v_n$ we only need to ensure condition (ii), which is equivalent to: for each $i\ge 3$, $|v_i,\{v_2,...,v_{i-1}\}|\ge |v_i,v_1|$ holds. 

We define permutation $v_3,...,v_n$ inductively. Suppose terms $v_2,...,v_{i-1}$ have been selected, where $3\le i\le n$. Let $Y'=\{v_2,...,v_{i-1}\}$. We prove that there exists a vertex in $Y\backslash Y'$, which we call $v_i$, such that $|v_i,Y'|\ge |v_i,v_1|$. Suppose otherwise that $|y,Y'|< |y,v_1|$ holds for all $y\in Y\backslash Y'$. Then $|Y\backslash Y', Y'|<|Y\backslash Y',v_1|$, implying $|Y',V\backslash Y'|=|Y',v_1|+|Y',Y\backslash Y'|<|Y',v_1|+|Y\backslash Y',v_1|=|X,Y|$, a contradiction. Thus $v_i$ can be selected, and this proves Claim 1.

\noindent Claim 2. If $|X|>1$ then there exist distinct $x,x'\in X$ with $|x',x|>|x',Y|$.

Suppose this is not the case. Then $|x',x|\le|x',Y|$ holds for all distinct $x,x'\in X$. Consequently, for any fixed $x\in X$, we have $|X\backslash x, x|\le |X\backslash x, Y|$, which implies $|x,V\backslash x|=|x,Y|+|x,X\backslash x|\le |x,Y|+ |X\backslash x, Y|=|X,Y|$. This contradicts the minimality of $|X|$ and thus Claim 2 is proved. 

Now we are ready to construct the required permutation for the case $|X|>1$. Let $x,x'$ be chosen as in Claim 2. Let $v_1=x$ and $v_{|Y|+2}=x'$. We will make vertices of $Y$ (specified below) $v_2,...,v_{|Y|+1}$ and vertices of $X$ (also specified below) $v_1,v_{|Y|+2},...,v_{|V|}$. Note that under this arrangement, conditions (i) and (iii) are satisfied. 

To determine a permutation of $Y$ we consider $G'$ obtained from $G$ by contracting $X$ into a single vertex, which we denote by $x^*$. Then $(\{x^*\},Y)$ is a minimum cut of $G'$. By Claim 1, vertices of $Y$ can be permuted to satisfy (ii). Note that satisfying (ii) in $G'$ and satisfying (ii) in $G$ are the same thing for vertices of $Y$. So we have obtained a required permutation for $Y$. 
Similarly, to determine a permutation of $X'=X\backslash \{x,x'\}$ we consider $G''$ obtained from $G$ by contracting $Y$ into a single vertex $y^*$ and also contracting $\{x,x'\}$ into a single vertex $x^*$. Again, $(X'\cup\{x^*\}, \{y^*\})$ is a minimum cut of $G''$. By Claim 1, vertices of $G''$ can be permuted to satisfy (ii), where if $u_1,....,u_{|X|}$ is the permutation of $V(G'')$ then $u_1=y^*$ and $u_2=x^*$ (as shown in the proof of Claim 1). It follows that setting $v_{|Y|+i}=u_i$ ($3\le i\le |X|$) results in a permutation of $V$ that satisfies (ii).
\end{proof}

\noindent This theorem suggests the tie-breading rule ({\bf R}) we mentioned above: 

({\bf R}) \ If $f_X=f_Y$ then $v_i$ goes to where $v_{i-1}$ went.

\noindent This is equivalent to the following rule.

({\bf R}) \ If $f_X=f_Y$ and $|X|=1$ then $Y=Y\cup\{v_i\}$; if  $f_X=f_Y$ and $|X|>1$ then $X=X\cup\{v_i\}$. 

\noindent Now we can formally describe our Greedy Algorithm \ref{alg: greedy min}.

\begin{algorithm}[h!]
\caption{ A Greedy proto-Algorithm $\mathsf{MinCut}$}
\begin{algorithmic}[1]
\STATE {\bf Input:} A graph $G$.
\STATE {\bf Output:} A cut of $G$.
\STATE Initialize: $X \leftarrow \{v_1\}$ , $Y \leftarrow \{v_2\}$ and $i \leftarrow 0$.
\WHILE{$2 < i \le n$}
\STATE $f_X = |v_i, X|$ and $f_Y = |v_i, Y|$
\IF{$f_X > f_Y$}
\STATE $X \leftarrow X \cup \{v_i\}$
\ELSIF{$f_X < f_Y$}
\STATE $Y \leftarrow Y \cup \{v_i\}$
\ELSE
\STATE Using tie breaker ({\bf R}) to decide if $X \leftarrow X \cup \{v_i\}$ or $Y \leftarrow Y \cup \{v_i\}$.\\ This decision is based on $G[\{v_1,\ldots,v_i\}]$
\ENDIF 
\ENDWHILE
\STATE $i \leftarrow i+1$
\end{algorithmic}
\label{alg: greedy min}
\end{algorithm}

\noindent Then the following is an immediate consequence of Theorem \ref{lem:mincut}

\begin{restatable}{theorem}{gmincutmore}\label{thm:greedyminc}
For the online mincut problem there exists a greedy algorithm $\mathbf A$ with the following property. For every loopless graph $G$ there exists a permutation of $V(G)$ such that when taking this permutation as its input $\mathbf A$ produces a minimum cut. 
\end{restatable}

\noindent To establish a similar result for the online $\mathsf{MaxCut}$ problem we need the following theorem.

\begin{restatable}{theorem}{gmaxcut}\label{thm:maxcutper}
Every loopless graph $G=(V,E)$ has a maximum cut $(X,Y)$ for which there exists a permutation $v_1,...,v_n$ of $V$ such that the following conditions are satisfied. For each $i\ge1$, let $X_i=X\cap\{v_1,...,v_i\}$ and $Y_i=Y\cap \{v_1,...,v_i\}$.\\ 
(i) $v_1\in X$ and $v_2\in Y$.\\ 
(ii) For every $i\ge 3$, if $v_i\in X$ then $|v_i,X_{i-1}|\le |v_i,Y_{i-1}|$ and if $v_i\in Y$ then $|v_i,Y_{i-1}|\le |v_i,X_{i-1}|$.\\ 
(iii) If $i\ge3$ and $|v_i,Y_{i-1}|= |v_i,X_{i-1}|$ then either $\{v_{i-1},v_i\}\subseteq X$ or $\{v_{i-1},v_i\}\subseteq Y$.
\end{restatable}

\begin{proof}
We first observe that there exists $k\in\{2,...,n\}$ for which there exists a maximum cut $(X,Y)$ and a permutation $v_1,...,v_n$ of $V$ such that (i) is satisfied and (ii-iii) are satisfied for all $i\in\{3,...,k\}$. To see this we only need to take $k=2$ and take any maximum cut $(X,Y)$, any $v_1\in X$, any $v_2\in Y$, and any permutation of $V$ starting with $v_1v_2$. 

Let us choose $k$ as large as possible under the above requirements. To prove the theorem we only need to show that $k=n$. Suppose on the contrary that $k<n$. Without loss of generality, let us assume $v_k\in Y$. If there exists $y\in Y\backslash Y_k$ with $|y,Y_k|\le |y,X_k|$ then setting $v_{k+1}=y$ (with the same maximum cut ($X,Y$)) would contradict the maximality of $k$. So  $|y,Y_k|>|y,X_k|$ holds for all $y\in Y\backslash Y_k$. Similarly, from the maximality of $k$ we deduce that $|x,Y_k|\le |x,X_k|$ holds for all $x\in X\backslash X_k$. Consequently, we must have $Y_k=Y$ and $|X\backslash X_k,X_k|=|X\backslash X_k,Y_k|$ because otherwise $(X_k\cup(Y\backslash Y_k), Y_k\cup(X\backslash X_k))$ would be a cut bigger than $(X,Y)$, a contradiction. But then replacing $(X,Y)$ with $(X_k, Y_k\cup(X\backslash X_k))$ and setting $v_{k+1}=x$ for any $x\in X\backslash X_k$ would contradict the maximality of $k$. Therefore, we must have $k=n$ and thus the theorem is proven.
\end{proof}

\noindent This theorem leads to the following.

\begin{restatable}{theorem}{gmaxcutmore}\label{thm:greedymaxcut}
For online MaxCut there exists a greedy algorithm $\mathbf A$ satisfying the following property. For every loopless graph $G$ there exists a permutation of $V(G)$ such that when taking this permutation as its input $\mathbf A$ produces a maximum cut.
\end{restatable}

\begin{algorithm}[h!]
\caption{ A Greedy Algorithm for $\mathsf{MaxCut}$}
\begin{algorithmic}[1]
\STATE {\bf Input:} A graph $G$.
\STATE {\bf Output:} A cut of $G$.
\STATE Initialize: $X \leftarrow \{v_1\}$ , $Y \leftarrow \{v_2\}$ and $i \leftarrow 0$.
\WHILE{$2 < i \le n$}
\STATE $f_X = |v_i, X|$ and $f_Y = |v_i, Y|$
\IF{$f_X < f_Y$}
\STATE $X \leftarrow X \cup \{v_i\}$
\ELSIF{$f_X > f_Y$}
\STATE $Y \leftarrow Y \cup \{v_i\}$
\ELSE
\STATE Put $v_i$ where $v_{i-1}$ went.
\ENDIF 
\ENDWHILE
\STATE $i \leftarrow i+1$
\end{algorithmic}
\label{alg: greedy max}
\end{algorithm}

\begin{proof}
We consider the greedy algorithm $\mathbf A$ given in Algorithm  \ref{alg: greedy max}.
\noindent To see that $\mathbf A$ satisfies the theorem, for any loopless graph $G$, let partition $(X,Y)$ and permutation $v_1....v_n$ be determined as in Theorem \ref{thm:maxcutper}. Then $\mathbf A$ produces exactly partition $(X,Y)$, which is a maximum cut, as required.
\end{proof}

Therefore, for MinCut and MaxCut problems, we established the existence of an algorithm $\mathbf A$ with $\min_\pi \mathbf A(G,\pi)=opt(G)$ for all $G$.

\subsection{Greedy property of submodular functions}

There are two related problems.  Suppose $G=(V,E)$ is a graph and $f: 2^V\to \mathbb R$ such that $f(X)$ is the number of edges between $X$ and $V\backslash X$. Then it is not difficult to verify that $f$ is a submodular function (defined below). So minimizing and maximizing a submodular function can be considered as a generalization of mincut and maxcut. However, for the corresponding online problems there is a subtle difference. For the online submodular problem, if $\Omega'$ is the set of currently revealed elements, then the algorithm can access  $f(X)$ for all $X$ contained in $\Omega$ (the domain of $f$). In contrast, if $V'$ is the set of currently revealed vertices and if $X\subseteq V'$, the algorithm cannot access  $f(X)$, it can only compute the number of edges between $X$ and $V'\backslash X$. 

Nevertheless, we developed a greedy type algorithm $\mathbf A^*$, which behaves very similar to $\mathbf A_{\min}$ and $\mathbf A_{\max}$. In particular, for every submodular function $f$, we constructed a permutation $\pi$ of $\Omega$ such that $\mathbf A^*(f,\pi)$ is a subset of $\Omega$ that maximizes $f$. In other words, we establish that $\min_\pi \mathbf A^*(f,\pi) = opt(f)$ holds for all submodular functions $f$. Finally, remark that no such $\mathbf A^*$ exists for minimizing a submodular function. 

\subsubsection{Definitions and Preliminaries}

Let $E$ be a finite set. A function $f: 2^E \to \mathbb R$ is called submodular if 
$$f(X)+f(Y)\ge f(X\cap Y) + f(X\cup Y)$$
holds for all $X,Y\subseteq E$. This can also be equivalently defined as 
$$f(X\cup\{y\})+f(X\cup\{z\})\ge f(X) + f(X\cup \{y,z\})$$
holds for all $X\subseteq E$ and all distinct $y,z\in E\backslash X$. The second definition is the same as saying that $f_e(X):=f(X\cup\{e\})-f(X)$ is a non-increasing function on $2^{E\backslash e}$ for all $e\in E$. In other words,
$$f(X\cup\{e\})-f(X) \ge f(Y\cup \{e\})-f(Y)$$
holds for all $e\in E$ and all $X\subseteq Y\subseteq E\backslash e$.

\noindent {\bf Remark}. In some papers like in \cite{chekuri2015streaming,hazan2012online} the following definition of a submodular function is used: 
$$f(X\cup\{e\})-f(X) \ge f(Y\cup \{e\})-f(Y)$$
holds for all $e\in E$ and all $X\subseteq Y\subseteq E$. It is easy to see that this condition is the same as saying $f_e$ is non-increasing on $2^E$ for all $e\in E$. Note that this definition is not what we have above since they have different domains. In fact, it is easy to see that $f$ satisfies this definition if and only if $f$ is submodular {\bf and} non-decreasing (i.e. $f(X)\le f(Y)$ holds for all $X\subseteq Y\subseteq E$).


\subsubsection{Online Submodular Maximization}

We consider the following online model for maximizing a submodular function $f$. First, we assume that $f$ is given by an oracle. That is, for any $X\subseteq E$, obtaining the value of $f(X)$ does not require extra work. The objective of the maximization problem is to find a {\it maximizer} $X$ of $f$, which means that $f(X)=\max\{f(Y): Y\subseteq E\}$. We assume that elements of $E$ are revealed one by one. At each step, when a new element $e$ is revealed, the algorithm has to decide if or not to place $e$ in $X$. This is an irrevocable decision. 
In the following we present an analogue of Theorem \ref{thm:greedymaxcut}. 

\begin{lemma}\label{lem:submod}
Let $f$ be a submodular function on a set $E$. If $X\subseteq E$ is a maximizer of $f$ then \\ 
\indent (i)  $f(Y)\le f(Y\cup\{x\})$ holds for every $x\in X$ and every $Y\subseteq X\backslash x$;\\ 
\indent (ii) if $X$ is a maximal (under inclusion) maximizer of $f$ then $f(Y\cup\{x\})<f(Y)$ holds for all \\ \makebox[33pt]{} $x\in E\backslash X$ and all $Y\subseteq E\backslash x$ with $Y\supseteq X$. 
\end{lemma}

\begin{proof}
Suppose (i) is false. Then $f(Y)>f(Y\cup\{x\})$ holds for some $x\in X$ and some $Y\subseteq X\backslash x$. From the submodularity of $f$ we deduce that $f(X\backslash x)+f(Y\cup\{x\})\ge f(X)+f(Y)$, which implies $f(X\backslash x)\ge f(X)+f(Y)-f(Y\cup\{x\})>f(X)$, contradicting the maximality of $f(X)$.

Suppose (ii) is false. Then $f(Y)\le f(Y\cup\{x\})$ holds for some $x\in E\backslash X$ and some $Y\subseteq E\backslash x$ with $Y\supseteq X$. Again, by the submodularity of $f$ we have $f(X\cup\{x\})+f(Y)\ge f(X) + f(Y\cup\{x\})$, and thus  $f(X\cup\{x\})\ge f(X)+ f(Y\cup\{x\})-f(Y)\ge f(X)$. This implies that $X\cup\{x\}$ is a maximizer of $f$, contradicting the choice of $X$.
\end{proof}

\begin{theorem} \label{thm:submax}
There exists an online greedy algorithm $\mathbf A$ with the following property. For any submodular function $f$ defined on a finite set $E$, that exists a permutation of $E$ such that when taking this permutation as its input $\mathbf A$ produces a maximizer of $f$.
\end{theorem}

\begin{algorithm}[h!]
\caption{ A greedy algorithm for maximizing a submodular  function online}
\begin{algorithmic}[1]
\STATE {\bf Input:} A submodular function $f$.
\STATE {\bf Output:} A subset of  $E$.
\STATE Initialize: $X \leftarrow \emptyset$ and  $i \leftarrow 0$.
\WHILE{$0 < i \le |E|$}
\IF{$f(X \cup \{e_i\}) \ge f(X)$}
\STATE $X \leftarrow X \cup \{e_i\}$
\ENDIF 
\STATE $i \leftarrow i+1$
\ENDWHILE
\end{algorithmic}
\label{alg: greedy sub}
\end{algorithm}
\begin{proof}
We consider the above greedy algorithm $\mathbf A$ as shown above. Let $e_1,e_2,...$ be the input sequence. 

\ifx false
Step 0: $X_0=\emptyset$. \\ 
\indent ... ... ... \medskip\\ 
\indent Step $i$: set $X_i=X_{i-1}\cup\{e_i\}$ if $f(X_{i-1}\cup\{e_i\})\ge f(X_{i-1})$, and set $X_i=X_{i-1}$ if otherwise.
\fi
\noindent To see that $\mathbf A$ satisfies the requirements, for any submodular function $f$ defined on $E$, let $X\subseteq E$ be a maximal maximizer of $f$. Consider a permutation of $E$ such that its first $|X|$ elements are from $X$. Then the result follows immediately from Lemma \ref{lem:submod}.
\end{proof}

\noindent{\sc Remarks.} 1. We called Theorem \ref{thm:submax} an analogue of Theorem \ref{thm:greedymaxcut} because they both deal with submodular functions and elements of the ground set are received one by one. However, we should point the main difference between them. In the $\mathsf{MaxCut}$ problem, if we use $f$ to denoted submodular function defined on $V(G)$, that is, $f(X)=|X,V\backslash X|$, we can see that at each iteration, we do not really now the values of $f(X)$. Instead, what we have is an approximation of it. \\ 
\indent 2. One may naturally ask for an analogue of Theorem \ref{thm:greedyminc}. But such a result do not exist, as shown by the following example. Consider a function $f$ defined on $E=\{x,y\}$, with $f(\emptyset)=0$, $f(x)=f(y)=1$, and $f(E)=-1$. This function is submodular: to see it we only need to verify inequalities $f(X)+f(Y)\ge f(X\cup Y)+f(X\cap Y)$ for incomparable subsets $X,Y$ of $E$. But this is clear since there is only such ineuqlity $f(x)+f(y)=2\ge -1=f(E)+f(\emptyset)$. \\ 
\indent Observe that $E$ is the unique minimizer for $f$, and upto symmetry, there is only one permutation $xy$ of $E$. However, the values of the corresponding subsets $\emptyset, \{x\}, \{x,y\}$ are $0, 1, -1$. Therefore, any greedy algorithm would return $\emptyset$ as the minimizer, which is not the real minimizer.


\bibliographystyle{plain}
\bibliography{ref}

\renewcommand{\thesubsection}{\Alph{subsection}}

\end{document}